\documentclass[11pt, onecolumn, draftcls]{IEEEtran}

\usepackage{graphicx}  
\usepackage{amssymb}
\usepackage{verbatim}
\usepackage{graphicx}
\usepackage{amsmath}

\newtheorem{thm}{Theorem}[section] 
\newtheorem{coro}[thm]{Corollary}
\newtheorem{lem}[thm]{Lemma} 
\newtheorem{defn}[thm]{Definition}

\begin{document}
\title{Shannon Theoretic Limits on Noisy Compressive Sampling}
\author{Mehmet~Ak\c{c}akaya
        ~and~Vahid~Tarokh
\thanks{M. Ak\c{c}akaya and V. Tarokh are with the School of Engineering and Applied Sciences, Harvard University, Cambridge, MA, 02138. (e-mails: \{akcakaya, vahid\}@deas.harvard.edu)}}
\markboth{DRAFT}{DRAFT}
\maketitle

\begin{abstract}
In this paper, we study the number of measurements required to recover a sparse signal in ${\mathbb C}^M$ with $L$ non-zero coefficients from compressed samples in the presence of noise. For a number of different recovery criteria, we prove that $O(L)$ (an asymptotically linear multiple of $L$) measurements are necessary and sufficient if $L$ grows linearly as a function of $M$. This improves on the existing literature that is mostly focused on variants of a specific recovery algorithm based on convex programming, for which $O(L\log(M-L))$ measurements are required. We also show that $O(L\log(M-L))$ measurements are required in the sublinear regime ($L = o(M)$).
\end{abstract}

\IEEEpeerreviewmaketitle

\section{Introduction}

\PARstart{L}{et} ${\mathbb C}$ denote the complex field and ${\mathbb C}^M$ the $M$-dimensional complex space. For any ${\bf x} \in {\mathbb C}^M$, let $||{\bf x}||_0$ denote the number of non-zero coefficients of ${\bf x}$. Whenever $||{\bf x}||_0 = L << M$, it is advantageous to measure a linear combination of the components of ${\bf x}$ as $${\bf y = Ax},$$ where ${\bf A}$ is an $N \times M$ measurement matrix.

A decoder can then recover ${\bf x}$ from the observed vector by solving the ${\cal L}_0$ minimization problem
$$\min ||{\bf x}||_0 \:\:\: \textrm{ s. t. } \:\:\: {\bf y = Ax}.$$

This data acquisition technique for sparse signals is called compressive sampling \cite{Candes-Tao, donoho}. However, the optimization problem for recovery is NP-hard to solve \cite{Tropp}. In this light, alternative solution methods have been studied in the literature. One such approach is the ${\cal L}_1$ regularization approach, where one solves 
$$\min ||{\bf x}||_1 \:\:\: \textrm{ s. t. } \:\:\: {\bf y = Ax},$$
and then establishes criteria under which the solution to this problem is also that of the ${\cal L}_0$ minimization problem. By considering certain classes of Gaussian and partial Fourier ensembles, Cand\`es and Tao showed in \cite{Candes-Tao} that this recovery problem could be solved for $L = O(M)$ with $N=O(L)$ as long as the observations are noiseless. Another strand of work considers solving the ${\cal L}_0$ recovery problem for a specific class of measurement matrices, such as the Vandermonde frames \cite{Ak-Tar1}.

In practice, however, all the measurements are noisy, i.e.
\begin{equation}
	{\bf y = Ax + n} \label{eq1}
\end{equation}
for some additive noise ${\bf n} \in {\mathbb C}^N$.
This motivates our work, where we study Shannon theoretic limits on the recovery of sparse signals in the presence of noise. More specifically, we are interested in the order of the number of measurements required, $N$ in terms of $L, M$. We consider the linear sparsity regime $M = \beta L$ for $\beta > 2$. It was shown in \cite{Ak-Tar1} that $\beta > 2$ is required even in the noiseless setting for the unique recovery of the signal.

Wainwright considered this problem with ${\bf n}$ being Gaussian noise in \cite{Wainwright}, and derived information theoretic limits on the noisy problem for a specific performance metric and a decoder that decodes to the closest subspace, showing that for the linear sparsity regime, the number of measurements required is also $O(L)$. 
In \cite{Wainwright2}, Wainwright studied the ${\cal L}_1$ constrained quadratic programming algorithm (LASSO) in the noisy setting and showed that in this case the number of measurements required is $N = O(L \log(M- L))$. Therefore there is a gap between what is achievable theoretically with an information theoretic decoder and what is achievable with a practical decoder based on ${\cal L}_1$ regularization. 
The total power of the signal, $$||{\bf x}||_2^2 = P$$ grows unboundedly as a function of $N$ according to the analysis in \cite{Wainwright}. The reason for this requirement is that at high dimensions, the performance metric in consideration is too stringent for an average case analysis. 

In this note, we consider various performance metrics, some of which are of more Shannon theoretic spirit. We use a decoder based on joint typicality. Although such a decoder may not be computationally feasible in practice, it enables us to characterize the performance limits on the sparse representation problem. Using this decoder, we first derive a result similar to that of \cite{Wainwright} for the same performance metric. For the other performance metrics that are more statistical in nature, we derive results stating that the number of required measurements is $O(L)$ and that $P$ does not have to grow with $N$.

The outline of this paper is given next. In Section \ref{sec:results}, we define the problem to be considered in this paper, establish the notation and performance metrics, and state our main results and their implications. Section \ref{sec:achieve} and Section \ref{sec:converse} provide the proofs for the theorems stated in Section \ref{sec:results}. In Section \ref{sec:sublinear}, we state analogous theorems for the sublinear sparsity regime, $L = o(M)$.


\section{Main Results} \label{sec:results}

We consider the compressive sampling of an unknown vector, ${\bf x} \in {\mathbb C}^M$. Let ${\bf x}$ have support ${\cal I} = \textrm{supp}({\bf x})$, where
$$\textrm{supp}({\bf x}) = \{i \mid x_i \neq 0 \}$$ with $||{\bf x}||_0 = |{\cal I}| = L = \big\lfloor\frac{1}{\beta} M\big\rfloor$, where $\beta > 2$. We also define \begin{equation} \mu({\bf x}) = \min_{i \in {\cal I}} |x_i|. \label{eq0} \end{equation} 

We consider the noisy model given in Equation (\ref{eq1}), where ${\bf n}$ is an additive noise vector with a complex circularly-symmetric Gaussian distribution with zero mean and covariance matrix $\nu^2 I_{N}$, i.e. ${\bf n} \sim {\cal N}_{C}(0, \nu^2 I_{N})$.
Due to the presence of noise, ${\bf x}$ cannot be recovered exactly. However, a sparse recovery algorithm outputs an estimate $\hat{{\bf x}}$ with $||\hat{{\bf x}}||_0 = L$. 
We consider three performance metrics for the estimate:
\begin{align}
\textrm{Error Metric 1: } \qquad \qquad p_1 (\hat{{\bf x}}, {\bf x}) &= {\mathbb I}\bigg(\big\{\hat{x}_i \neq 0 \:\: \forall i \in {\cal I} \big\} \cap \big\{\hat{x}_j = 0 \:\: \forall j \notin {\cal I} \big\} \bigg) \\
\textrm{Error Metric 2: } \qquad \qquad p_2 (\hat{{\bf x}}, {\bf x}) &= {\mathbb I}\bigg( \frac{|\{i \mid \hat{x}_i \neq 0 \} \cap {\cal I}|}{|{\cal I}|} > 1 - \alpha \bigg) \\ 
\textrm{Error Metric 3: } \qquad \qquad p_3 (\hat{{\bf x}}, {\bf x}) &= {\mathbb I}\Bigg( \sum_{k \in \{i \mid \hat{x}_i \neq 0 \} \cap {\cal I}} |x_k|^2 > (1 - \gamma)P \Bigg) 
\end{align}
where ${\mathbb I}(\cdot)$ is the indicator function and $\alpha, \gamma \in (0,1)$. 

Error Metric 1 is referred to as the $0$-$1$ loss metric, and it is the one considered by Wainwright \cite{Wainwright}. Error Metric 2 is a statistical extension of Error Metric 1, and considers the recovery of most of the subspace information of ${\bf x}$. 
Error Metric 3 is directly from Shannon Theory and characterizes the recovery of most of the energy of ${\bf x}$.

Consider a sequence of vectors, $\{{\bf x}^{(M)}\}_M$ such that ${\bf x}^{(M)} \in {\mathbb C}^M$ with ${\cal I}^{(M)} = \textrm{supp}({\bf x}^{(M)})$, where $|{\cal I}^{(M)}| = L^{(M)} = \big\lfloor \frac{1}{\beta} M \big\rfloor$. For ${\bf x}^{(M)}$, we will consider an ensemble of $N \times M$ Gaussian measurement matrices, ${\bf A}^{(M)}$, where $N$ is a function of $M$. Since the dependence of $L^{(M)}, {\cal I}^{(M)}$ and ${\bf A}^{(M)}$ on $M$ is implied by the vector ${\bf x}^{(M)}$, we will omit the superscript for brevity, and denote the support of ${\bf x}^{(M)}$ by ${\cal I}$, its size by $L$ and any measurement matrix from the ensemble by ${\bf A}$, whenever there is no ambiguity.

A decoder, ${\cal D}(\cdot)$ will output a set of indices, ${\cal D}({\bf y})$. For a specific decoder, we consider the average probability of error, averaged over all \emph{Gaussian measurement matrices}, ${\bf A}$ with the $(i,j)^{\textrm{th}}$ term $a_{i,j} \sim {\cal N}_{C}(0, 1)$:
\begin{align}
	p_{\textrm{err}} ({\cal D} | {\bf x}^{(M)}) &= {\mathbb E}_{\bf A}\big(p_{\textrm{err}}({\bf A}  | {\bf x}^{(M)})\big), 
\end{align}
where $p_{\textrm{err}}({\bf A} | {\bf x}^{(M)}) = {\mathbb P}({\cal D}({\bf y}) \neq {\cal I})$ for ${\bf y = Ax}^{(M)} + {\bf n}$ and ${\mathbb P}(\cdot)$ is the probability measure. 

We say a decoder achieves \emph{asymptotic reliable} sparse recovery if $p_{\textrm{err}} ({\cal D}|{\bf x}^{(M)}) \to 0$ as $M \to \infty$. Similarly we say asymptotic reliable sparse recovery is not possible if $p_{\textrm{err}} ({\cal D} | {\bf x}^{(M)})$ stays bounded away from 0 as $M \to \infty$. 

We also use the notation $$f(x) \succ g(x)$$ for either $f(x) = g(x) = 0$ or for non-decreasing non-negative functions $f(x)$ and $g(x)$, if $\exists \: x_0$ such that for all $x > x_0$, $$\frac{f(x)}{g(x)} > 1.$$
Similarly we say $f(x) \prec g(x)$ if $g(x) \succ f(x)$.

\begin{thm} \label{thm1} (Achievability for Error Metric 1) Let a sequence of sparse vectors, $\{{\bf x}^{(M)} \in {\mathbb C}^M\}_M$ with $||{\bf x}^{(M)}||_0 = L = \big\lfloor\frac{1}{\beta} M\big\rfloor$, where $\beta > 2$ be given. Then asymptotic reliable recovery is possible for $\{{\bf x}^{(M)}\}$ with respect to Error Metric 1 if $\frac{L {\mu}^4({\bf x}^{(M)})}{\log L} \to \infty$ as $L \to \infty$ and
\begin{equation}	
	  N \succ  C_1 \: L
\end{equation}
for some constant $C_1>1$ that depends only on $\beta$, ${\mu}({\bf x}^{(M)})$ and $\nu$.
\end{thm}

\begin{proof} The proof is given in Section \ref{sssec1}. \end{proof}
\begin{coro} \label{coro1} Let the conditions of Theorem \ref{thm1} be satisfied. Then for any Gaussian measurement matrix, ${\bf A}$, and for Error Metric 1, $-\log{\mathbb P}(p_{\textrm{err}}({\bf A}|{\bf x}^{(M)}) \geq \xi) \slash \log L \to \infty$ as $L \to \infty$ for any $\xi \in (0,1]$.
\end{coro}
\begin{proof} Markov's Inequality implies $${\mathbb P}(p_{\textrm{err}}({\bf A}|{\bf x}^{(M)}) \geq \xi) \leq \frac{{\mathbb E}_{{\bf A}}(p_{\textrm{err}}({\bf A}|{\bf x}^{(M)}))}{\xi} = \frac{p_{\textrm{err}} ({\cal D}|{\bf x}^{(M)})}{\xi}.$$

As shown in the proof of Theorem \ref{thm1}, $-\log p_{\textrm{err}} ({\cal D}|{\bf x}^{(M)}) \slash \log L \to \infty$ as $L \to \infty$, yielding the desired result. \end{proof}

\begin{thm} \label{thm2} (Converse for Error Metric 1) 
 Let a sequence of sparse vectors, $\{{\bf x}^{(M)} \in {\mathbb C}^M\}_M$ with $||{\bf x}^{(M)}||_0 = L = \big\lfloor\frac{1}{\beta} M\big\rfloor$, where $\beta > 2$ be given. Then asymptotic reliable recovery is not possible for $\{{\bf x}^{(M)}\}$ with respect to Error Metric 1 if 
\begin{equation}	
	  N \prec  C_2 \frac{L}{\log P}
\end{equation}
for some constant $C_2>0$ that depends only on $\beta, P$ and $\nu$.
\end{thm}
\begin{proof} The proof is given in Section \ref{sssec4}. \end{proof}

\begin{coro} \label{coro4}
Let a sequence of sparse vectors, $\{{\bf x}^{(M)} \in {\mathbb C}^M\}_M$ with $||{\bf x}^{(M)}||_0 = L = \big\lfloor\frac{1}{\beta} M\big\rfloor$, where $\beta > 2$ be given. 
Then for $\xi > 0$, for any Gaussian measurement matrix, ${\bf A}$, and for Error Metric 1,  ${\mathbb P}\big(p_{err}({\bf A}|{\bf x}^{(M)}) \to 1)$ 
goes to 1 exponentially fast as a function of $M$ if $N \prec \hat{C}_2 \frac{L}{\log P}$, where $\hat{C}_2 < C_2$ is a positive constant that depends only on $\beta, P, \nu$ and $\xi $.
\end{coro}
\begin{proof} The proof is given in Section \ref{sssec4}. \end{proof}

\begin{thm} \label{thm3} (Achievability for Error Metric 2)  Let a sequence of sparse vectors, $\{{\bf x}^{(M)} \in {\mathbb C}^M\}_M$ with $||{\bf x}^{(M)}||_0 = L = \big\lfloor\frac{1}{\beta} M\big\rfloor$, where $\beta > 2$ be given such that $L\mu^2({\bf x}^{(M)})$ and $P$ are constant. Then asymptotic reliable recovery is possible for $\{{\bf x}^{(M)}\}$ with respect to Error Metric 2 if
\begin{equation}	
	  N \succ  C_3 \: L
\end{equation}
for some constant $C_3>1$ that depends only on $\alpha$, $\beta$, ${\mu}({\bf x}^{(M)})$ and $\nu$. 
\end{thm}
\begin{proof} The proof is given in Section \ref{sssec2}. \end{proof}

\begin{coro} \label{coro2} Let the conditions of Theorem \ref{thm3} be satisfied. Then for any Gaussian measurement matrix, ${\bf A}$, and for Error Metric 2, ${\mathbb P}(p_{\textrm{err}}({\bf A}|{\bf x}^{(M)}) > \xi)$ is exponentially decaying to zero as a function of $M$ for any $\xi \in (0,1]$.
\end{coro}
\begin{proof} As shown in the proof of Theorem \ref{thm3}, $p_{\textrm{err}} ({\cal D}|{\bf x}^{(M)})$ decays exponentially fast in $M$. Applying Markov's Inequality, 
yields the desired result. \end{proof}

\begin{thm} \label{thm4} (Converse for Error Metric 2) 
 Let a sequence of sparse vectors, $\{{\bf x}^{(M)} \in {\mathbb C}^M\}_M$ with $||{\bf x}^{(M)}||_0 = L = \big\lfloor\frac{1}{\beta} M\big\rfloor$, where $\beta > 2$ be given such that $P$ is constant. Then asymptotic reliable recovery is not possible for $\{{\bf x}^{(M)} \}$ with respect to Error Metric 2 if 
\begin{equation}	
	  N \prec C_4 L
\end{equation}
for some constant $C_4 \geq 0$ that depends only on $\alpha, \beta, P$ and $\nu$.
\end{thm}
\begin{proof} The proof is given in Section \ref{sssec5}. \end{proof}

\begin{coro} \label{coro5}  
Let a sequence of sparse vectors, $\{{\bf x}^{(M)} \in {\mathbb C}^M\}_M$ with $||{\bf x}^{(M)}||_0 = L = \big\lfloor\frac{1}{\beta} M\big\rfloor$, where $\beta > 2$ be given such that $P$ is constant. 
Then for $\xi > 0$, for any Gaussian measurement matrix, ${\bf A}$, and for Error Metric 2,  ${\mathbb P}\big(p_{err}({\bf A}|{\bf x}^{(M)}) \to 1)$ 
goes to 1 exponentially fast as a function of $M$ if $N \prec \hat{C}_4 L$, where $\hat{C}_4 \leq C_4$ is a non-negative constant that depends only on $\alpha, \beta, P, \nu$ and $\xi $.
\end{coro}
\begin{proof} The proof is analogous to the proof of Corollary \ref{coro4}. \end{proof}

\begin{thm} \label{thm5} (Achievability for Error Metric 3) Let a sequence of sparse vectors, $\{{\bf x}^{(M)} \in {\mathbb C}^M\}_M$ with $||{\bf x}^{(M)}||_0 = L = \big\lfloor\frac{1}{\beta} M\big\rfloor$, where $\beta > 2$ be given such that $P$ is constant. Then asymptotic reliable recovery is possible for $\{{\bf x}^{(M)}\}$ with respect to Error Metric 3 if
\begin{equation}	
	  N \succ  C_5 \: L
\end{equation}
for some constant $C_5>1$ that depends only on $\beta, \gamma$, $P$ and $\nu$. 
\end{thm}
\begin{proof} The proof is given in Section \ref{sssec3}. \end{proof}

\begin{coro} \label{coro3} Let the conditions of Theorem \ref{thm5} be satisfied. Then for any Gaussian measurement matrix, ${\bf A}$, and for Error Metric 3, ${\mathbb P}(p_{\textrm{err}}({\bf A}|{\bf x}^{(M)}) > \xi)$ is exponentially decaying to zero as a function of $M$ for any $\xi \in (0,1]$.
\end{coro}
\begin{proof} The proof is analogous to the proof of Corollary \ref{coro2}. \end{proof}

\begin{thm} \label{thm6} (Converse for Error Metric 3) 
 Let a sequence of sparse vectors, $\{{\bf x}^{(M)} \in {\mathbb C}^M\}_M$ with $||{\bf x}^{(M)}||_0 = L = \big\lfloor\frac{1}{\beta} M\big\rfloor$, where $\beta > 2$ be given such that $P$ is constant and the non-zero terms decay to zero at the same rate. Then asymptotic reliable recovery is not possible for $\{{\bf x}^{(M)}\}$ with respect to Error Metric 3 if 
\begin{equation}	
	  N \prec C_6 L
\end{equation}
for some constant $C_6 \geq 0$ that depends only on $\beta, \gamma, P, \mu({\bf x}^{(M)})$ and $\nu$.
\end{thm}
\begin{proof} The proof is given in Section \ref{sssec6}. \end{proof}

\begin{coro} \label{coro6} 
Let a sequence of sparse vectors, $\{{\bf x}^{(M)} \in {\mathbb C}^M\}_M$ with $||{\bf x}^{(M)}||_0 = L = \big\lfloor\frac{1}{\beta} M\big\rfloor$, where $\beta > 2$ be given such that $P$ is constant and the non-zero terms decay to zero at the same rate. 
Then for $\xi > 0$, for any Gaussian measurement matrix, ${\bf A}$, and for Error Metric 3,  ${\mathbb P}\big(p_{err}({\bf A}|{\bf x}^{(M)}) \to 1)$ 
goes to 1 exponentially fast as a function of $M$ if $N \prec \hat{C}_6 L$, where $\hat{C}_6 \leq C_6$ is a non-negative constant that depends only on $\beta, \gamma, P, \mu({\bf x}^{(M)}), \nu$ and $\xi $.
\end{coro}
\begin{proof} The proof is analogous to the proof of Corollary \ref{coro4}. \end{proof}

\subsection{Discussion of The Results}
 
Theorem \ref{thm1} implies that for Error Metric 1, $O(L)$ measurements are sufficient for asymptotic reliable sparse recovery. There is a clear gap between this number of measurements and $O(L\log(M-L))$ measurements required by ${\cal L}_1$ constrained quadratic programming \cite{Wainwright2}. 
In this proof, it is required that $\frac{L \mu^4({\bf x}^{(M)})}{\log L} \to \infty$ as $L \to \infty$, which 
implies that $P$ grows without bound as a function of $N$.

Theorems \ref{thm3} and \ref{thm5} show that for Error Metrics 2 and 3, the number of required measurements to achieve asymptotic reliable sparse recovery is $N = O(L)$. In this case $P$ remains constant, which is a much less stringent requirement than that of Theorem \ref{thm1}.
Converses to these theorems are established in Theorems \ref{thm2}, \ref{thm4} and \ref{thm6}, which demonstrate that $O(L)$ measurements are asymptotically necessary.

Finally we note that Corollaries \ref{coro2} and \ref{coro3} imply that with overwhelming probability (i.e. the probability goes to $1$ exponentially fast as a function of $M$) a given $N \times M$ Gaussian measurement matrix ${\bf A}$ can be used for asymptotic reliable sparse recovery (respectively for Error Metrics 2 and 3) as long as $N = O(L)$. Similarly Corollaries \ref{coro5} and \ref{coro6} prove that a given Gaussian matrix ${\bf A}$ will have $p_{err}({\bf A}|{\bf x}^{(M)}) \to 1$ (respectively for Error Metrics 2 and 3) with overwhelming probability as long as the number of measurements is less than specified constant multiples of $L$. Corollaries \ref{coro1} and \ref{coro4} are similar in nature.

\section{Achievability Proofs} \label{sec:achieve}

\subsection{Notation}
Let ${\bf a}_i$ denote the $i^{\textrm{th}}$ column of ${\bf A}$. For the measurement matrix ${\bf A}$, we define ${\bf A}_{\cal J}$ to be the matrix whose columns are $\{{\bf a}_j \: : \: j \in {\cal J}\}$. For any given matrix ${\bf B}$, we define $\Pi_{{\bf B}}$ to be the orthogonal projection matrix onto the subspace spanned by the columns of ${\bf B}$, i.e. ${\bf \Pi}_{{\bf B}} = {\bf B(B^*B)^{-1}B^*}$. Similarly, we define ${\bf \Pi}_{{\bf B}}^{\perp}$ to be the projection matrix onto the orthogonal complement of this subspace, i.e. ${\bf \Pi}_{{\bf B}}^{\perp} = {\bf I} - {\bf \Pi}_{{\bf B}}$.

\subsection{Joint Typicality}

In our analysis, we will use Gaussian measurement matrices and a suboptimal decoder based on joint typicality, as defined below:

\begin{defn} (Joint Typicality) We say an $N \times 1$ noisy observation vector, ${\bf y = Ax + n}$ and a set of indices ${\cal J} \subset \{1, 2, \dots, M \}$, with $|{\cal J}| = L$, are $\delta$-jointly typical if rank$({\bf A}_{\cal J}) = L$ and 
\begin{equation} 
	\bigg| \frac1N ||{\bf \Pi}_{{\bf A}_{\cal J}}^{\perp} {\bf y}||^2 - \frac{N-L}{N} \nu^2 \bigg| < \delta,
\end{equation}
where ${\bf n} \sim {\cal N}_C(0, \nu^2 I_N)$, the $(i,j)^{\textrm{th}}$ entry of ${\bf A}$, $a_{ij} \sim {\cal N}_C(0,1)$, and $||{\bf x}||_0 = L$.
\end{defn}

\begin{lem} \label{lem0} For an index set ${\cal I}\subset \{1,2,\dots, M\}$ with $|{\cal I}| = L$, $${\mathbb P}(\textrm{rank}({\bf A}_{\cal I}) < L) = 0.$$
\end{lem}

\begin{lem} \label{lem1} $\:$ \\ \vspace{-0.9cm}
 \begin{itemize}  \item Let ${\cal I} = \textrm{supp}({\bf x})$ and assume (without loss of generality) that rank$({\bf A}_{\cal I}) = L$. Then for $\delta > 0$, 
\begin{equation}
  {\mathbb P}\bigg( \Big| \frac1N ||{\bf \Pi}_{{\bf A}_{\cal I}}^{\perp} {\bf y}||^2 - \frac{N-L}{N} \nu^2  \Big| > \delta  \bigg) \leq 2\exp\Bigg(-\frac{\delta^2}{4\nu^4} \frac{N^2}{N-L+\frac{2\delta}{\nu^2}N} \Bigg).
\end{equation}

\item Let ${\cal J}$ be an index set such that $|{\cal J}| = L$ and $|{\cal I} \cap {\cal J}| = K < L$, where ${\cal I} = \textrm{supp}({\bf x})$ and assume that rank$({\bf A}_{\cal J}) = L$. Then ${\bf y}$ and ${\cal J}$ are $\delta$-jointly typical with probability
\begin{equation}
{\mathbb P} \bigg( \Big| \frac1N ||{\bf \Pi}_{{\bf A}_{\cal J}}^{\perp} {\bf y}||^2 - \frac{N-L}{N} \nu^2  \Big| < \delta  \bigg)  \leq \exp\Bigg(- \frac{N-L}{4} \bigg(\frac{\sum_{k \in {\cal I} \backslash {\cal J}} |x_k|^2 - \delta'}{\sum_{k \in {\cal I} \backslash {\cal J}} |x_k|^2 + \nu^2}\bigg)^2 \Bigg),  
\end{equation}
where $$\delta' = \delta \frac{N}{N-L}.$$
\end{itemize}
\end{lem}
\medskip

\begin{proof}
We first note that for $${\bf y =Ax + n} = \sum_{i \in {\cal I}} x_i {\bf a}_i + {\bf n},$$ we have
$${\bf \Pi}_{{\bf A}_{\cal I}}^{\perp} {\bf y} = {\bf \Pi}_{{\bf A}_{\cal I}}^{\perp} {\bf n},$$ and $${\bf \Pi}_{{\bf A}_{\cal J}}^{\perp} {\bf y} = {\bf \Pi}_{{\bf A}_{\cal J}}^{\perp} \bigg(\sum_{i \in {\cal I} \backslash {\cal J}} x_i {\bf a}_i + {\bf n}\bigg).$$

Furthermore ${\bf \Pi}_{{\bf A}_{\cal I}}^{\perp} = {\bf U}_{\cal I} {\bf D}{\bf U}_{\cal I}^{\dagger}$, where ${\bf U}_{\cal I}$ is a unitary matrix that is a function of $\{{\bf a_i} : i \in {\cal I}\}$ (and independent of ${\bf n}$). ${\bf D}$ is a diagonal matrix with $N-L$ diagonal entries equal to 1, and the rest equal to 0. It is easy to see that
$$||{\bf \Pi}_{{\bf A}_{\cal I}}^{\perp} {\bf y}||^2 = ||{\bf Dn'}||^2,$$
where ${\bf n'}$ has i.i.d. entries with distribution ${\cal N}_C(0,\nu^2)$. Without loss of generality, we may assume the non-zero entries of ${\bf D}$ are on the first $N-L$ diagonals, thus 
$$||{\bf Dn'}||^2 = |n'_1|^2 + \dots + |n'_{N-L}|^2.$$

Similarly, ${\bf \Pi}_{{\bf A}_{\cal J}}^{\perp} = {\bf U}_{\cal J} {\bf D} {\bf U}_{\cal J}^{\dagger}$, where ${\bf U}_{\cal J}$ is a unitary matrix that is a function of $\{{\bf a_j} : j \in {\cal J}\}$ $\big($and independent of ${\bf n}$ and $\{ {\bf a}_i : i \in {\cal I} \backslash {\cal J} \} \big)$ and ${\bf D}$ is as discussed above. Thus ${\bf a'}_i = {\bf U}_{\cal J}^{\dagger} {\bf a}_i$ has i.i.d. entries with distribution ${\cal N}_C(0,1)$ for all $i \in {\cal I} \backslash {\cal J}$. It is easy to see that ${\bf n''} = {\bf U}_{\cal J}^{\dagger} {\bf n}$ also has i.i.d. entries with ${\cal N}_C(0,\nu^2)$. Thus
$$||{\bf \Pi}_{{\bf A}_{\cal J}}^{\perp} {\bf y}||^2 = ||{\bf Dw}||^2 = |w_1|^2 + \dots + |w_{N-L}|^2,$$
where $w_i$ are i.i.d. with distribution ${\cal N}_C(0, \sigma_{\cal J}^2)$, where $$\sigma_{\cal J}^2 = \sum_{k \in {\cal I} \backslash {\cal J}} |x_k|^2 +\nu^2.$$

Let $\Omega_1 = \frac{||{\bf Dn'}||^2}{\nu^2}$ and $\Omega_2 = \frac{||{\bf Dw}||^2}{\sigma_{\cal J}^2}$. We note that both  $\Omega_1$ and $\Omega_2$ are chi-square random variables with $(N-L)$ degrees of freedom. Thus to bound these probabilities, we must bound the tail of a chi-square random variable. We have,
\begin{align}
	{\mathbb P}  \Bigg( \bigg| \frac1N ||{\bf \Pi}_{{\bf A}_{\cal I}}^{\perp} {\bf y}||^2 &- \frac{N-L}{N} \nu^2  \bigg| > \delta  \Bigg) = 	{\mathbb P} \Bigg( \bigg|\Omega_1 - (N-L)\bigg| > \frac{\delta}{\nu^2} N  \Bigg)  \nonumber \\
	   &= {\mathbb P}\bigg(\Omega_1 - (N-L) < -\frac{\delta}{\nu^2} N \bigg) + {\mathbb P}\bigg(\Omega_1 - (N-L) > \frac{\delta}{\nu^2} N \bigg),	\label{bound-omega1}
\end{align}

and 
\begin{align}
	{\mathbb P}  \Bigg( \bigg| \frac1N ||{\bf \Pi}_{{\bf A}_{\cal J}}^{\perp} {\bf y}||^2 &- \frac{N-L}{N} \nu^2  \bigg| < \delta  \Bigg) = {\mathbb P} \Bigg( \bigg|\Omega_2 - (N-L) \frac{\nu^2}{\sigma_{\cal J}^2}  \bigg| < \frac{\delta}{{\sigma_{\cal J}^2}} N  \Bigg) \nonumber\\
	   &\leq {\mathbb P}\bigg(\Omega_2 - (N-L) < -(N-L)\bigg(1 - \frac{\nu^2}{\sigma_{\cal J}^2}\bigg) + \frac{\delta}{{\sigma_{\cal J}^2}} N \bigg)	\label{bound-omega2}
\end{align}

For a chi-square random variable, $\Omega$ with $(N-L)$ degrees of freedom \cite{Birge, Laurent},
\begin{equation}	\label{eq-dev1}
	{\mathbb P}\Big(\Omega - (N-L) \leq -2 \sqrt{(N-L)\lambda}\Big) \leq e^{-\lambda},
\end{equation}
and 
\begin{equation} \label{eq-dev2}
	{\mathbb P}\Big(\Omega - (N-L) \geq 2 \sqrt{(N-L)\lambda} + 2 \lambda\Big) \leq e^{-\lambda}.
\end{equation}
By replacing $\Omega = \Omega_1$ and
\begin{equation}
	\lambda = \bigg(\frac{\delta N}{2\nu^2\sqrt{N-L}}\bigg)^2 \nonumber
\end{equation} 
in Equation (\ref{eq-dev1}) and
\begin{equation}
	\lambda = \frac14 \bigg(\sqrt{N-L + \frac{2\delta}{\nu^2} N} - \sqrt{N-L} \bigg)^2 \geq \frac{\delta^2}{4 \nu^4} \frac{N^2}{N - L +  \frac{2\delta}{\nu^2} N } \nonumber
\end{equation} 
in Equation (\ref{eq-dev2}), we obtain using Equation (\ref{bound-omega1})
\begin{align}
  {\mathbb P}\bigg( \Big| \frac1N ||{\bf \Pi}_{{\bf A}_{\cal I}}^{\perp} {\bf y}||^2 - \frac{N-L}{N} \nu^2  \Big| > \delta  \bigg) &\leq \exp\Bigg(-\frac{\delta^2}{4\nu^4} \frac{N^2}{N-L} \Bigg) + \exp \Bigg(-\frac{\delta^2}{4\nu^4} \frac{N^2}{N-L +  \frac{2\delta}{\nu^2} N } \Bigg) \nonumber \\
  &\leq 2\exp \Bigg(-\frac{\delta^2}{4\nu^4} \frac{N^2}{N-L +  \frac{2\delta}{\nu^2} N } \Bigg) . \nonumber
\end{align} 

Similarly by replacing $\Omega = \Omega_2$ and 
\begin{align}
	\lambda &= \bigg(\frac{\sqrt{N-L}}{2}\bigg(1 - \frac{\nu^2}{\sigma_{\cal J}^2} \bigg) - \frac{\delta}{\sigma_{\cal J}^2} \frac{N}{2\sqrt{(N-L)}} \bigg)^2 \nonumber \\
	&= \bigg(\frac{\sqrt{N-L}}{2}\bigg(1 - \frac{\nu^2}{\sigma_{\cal J}^2} - \frac{\delta}{\sigma_{\cal J}^2}\frac{N}{N-L} \bigg) \bigg)^2 \nonumber
\end{align} 
in Equation (\ref{eq-dev1}), we obtain using Equation (\ref{bound-omega2})
\begin{align}
	{\mathbb P} \bigg( \Big| \frac1N ||{\bf \Pi}_{{\bf A}_{\cal J}}^{\perp} {\bf y}||^2 - \frac{N-L}{N} \nu^2  \Big| < \delta  \bigg) &\leq \exp\bigg(- \frac{N-L}{4} \Big(\frac{\sigma_{\cal J}^2 - \nu^2 - \delta'}{ \sigma_{\cal J}^2}\Big)^2\bigg) \nonumber \\
		&= \exp\Bigg(- \frac{N-L}{4} \bigg(\frac{\sum_{k \in {\cal I} \backslash {\cal J}} |x_k|^2 - \delta'}{\sum_{k \in {\cal I} \backslash {\cal J}} |x_k|^2 + \nu^2}\bigg)^2 \Bigg).  \nonumber
\end{align} 
\end{proof}

\subsection{Proofs of Theorems For Different Error Metrics}

We define the event 
\begin{equation}
	E_{\cal J} = \{{\bf y} \textrm{ and } {\cal J} \textrm{ are } \delta \textrm{-jointly typical } \} \nonumber
\end{equation}
for all ${\cal J} \subset \{1, \dots, M\}, \: |{\cal J}| = L$. 

We also define the error event
$$ E_0 = \{\textrm{rank}({\bf A}_{\cal I}) < L \},$$
which results in an order reduction in the model, and implies that the decoder is looking through subspaces of incorrect dimension. By Lemma \ref{lem0}, we have ${\mathbb P} (E_0) = 0$.

Since the relationship between $M$ and ${\bf x}^{(M)}$ is implicit in the following proofs, we will suppress the superscript and just write ${\bf x}$ for brevity.

\subsubsection{Proof of Theorem \ref{thm1} (Error Metric 1)} \label{sssec1}
Clearly the decoder fails if  $E_0$ or $E_{\cal I}^C$ occur or when one of $E_{\cal J}$ occurs for ${\cal J} \neq {\cal I}$. Thus
\begin{align}
p_{\textrm{err}} ({\cal D} | {\bf x}) &= {\mathbb P} \big(E_0 \cup E_{\cal I}^C  \bigcup_{{\cal J}, {\cal J} \neq {\cal I}, |{\cal J}| = L} E_{\cal J} \big) \nonumber \\
		&\leq {\mathbb P} (E_{\cal I}^C) + \sum _{{\cal J}, {\cal J} \neq {\cal I}, |{\cal J}| = L} {\mathbb P}(E_{\cal J}) \nonumber
\end{align}

We let $N = (4C_0 + 1)L$ where $C_0 > 2 + \log(\beta - 1)$ is a constant. Thus $\delta' = \frac{4C_0 +1}{4C_0} \delta = C'_0 \delta$ with $C'_0 > 1$. Also by the statement of Theorem \ref{thm1}, we have $L \mu^4({\bf x})$ grows faster than $\log L$. We note that this requirement is milder than that of \cite{Wainwright}, where the growth requirement is on $\mu^2({\bf x})$ rather than $\mu^4({\bf x})$. Since the decoder needs to distinguish between even the smallest non-overlapping coordinates, we let $\delta' = \zeta \mu^2({\bf x})$ for $0 < \zeta < 1$. For computational convenience, we will only consider $2/3 < \zeta <1 $.

By Lemma \ref{lem1},
\begin{equation}
 {\mathbb P}(E_{\cal I}^C) \leq 
 																2 \exp\Bigg(-\frac{\zeta^2 C_0}{\nu^2} \frac{L\mu^4({\bf x})}{\nu^2 + 2 \zeta \mu^2({\bf x})}\Bigg) \nonumber
\end{equation}
and by the condition on the growth of $\mu({\bf x})$, the term in the exponent grows faster than $\log L$. Thus ${\mathbb P}(E_{\cal I}^C)$ goes to $0$ faster than $\exp(-\log L)$.

Again by Lemma \ref{lem1}, for ${\cal J}$ with $|{\cal I} \cap {\cal J}|= K$, 
\begin{align}
	{\mathbb P}(E_{\cal J}) &\leq \exp\Bigg(- \frac{N-L}{4} \bigg(\frac{\sum_{k \in {\cal I} \backslash {\cal J}} |x_k|^2 - \delta'}{\sum_{k \in {\cal I} \backslash {\cal J}} |x_k|^2 + \nu^2}\bigg)^2 \Bigg)   \nonumber
\end{align}

Since $\sum_{k \in {\cal I} \backslash {\cal J}} |x_k|^2 \geq (L-K){\mu}^2({\bf x})$, we have
\begin{align}
	{\mathbb P}(E_{\cal J}) &\leq \exp\Bigg(- \frac{N-L}{4} \bigg(\frac{(L-K){\mu}^2({\bf x}) - \delta'}{(L-K){\mu}^2({\bf x}) + \nu^2}\bigg)^2 \Bigg), \label{eq17}
\end{align}
where ${\mu}({\bf x})$ is defined in Equation (\ref{eq0}).

The condition of Theorem \ref{thm1} on $\mu({\bf x})$ implies that ${\mathbb P}(E_{\cal J}) \to 0$ for all $K$. We note that this condition also implies $P \to \infty$ as $N$ grows without bound. This is due to the stringent requirements imposed by Error Metric 1 in high-dimensions.

By a simple counting argument, the number of subsets ${\cal J}$ that overlaps ${\cal I}$ in $K$ indices (and such that rank$({\bf A}_{\cal J}) = L$) is upper-bounded by $${L \choose K} {M-L \choose L-K}.$$ Thus  
\begin{align}
	p_{\textrm{err}} ({\cal D}&|{\bf x}) \leq 2 \exp\Bigg(-\frac{\zeta^2 C_0}{\nu^2} \frac{L\mu^4({\bf x})}{\nu^2 + 2 \zeta \mu^2({\bf x})}\Bigg) \nonumber \\ &\quad\quad\quad\quad + \sum_{K=0}^{L-1} {L \choose L -K} {M-L \choose L-K} \exp\Bigg(- \frac{N-L}{4} \bigg(\frac{(L-K){\mu}^2({\bf x}) - \delta'}{(L-K){\mu}^2({\bf x}) + \nu^2}\bigg)^2 \Bigg) \nonumber \\
	&= 2 \exp\Bigg(-\frac{\zeta^2 C_0}{\nu^2} \frac{L\mu^4({\bf x})}{\nu^2 + 2 \zeta \mu^2({\bf x})}\Bigg)  + \sum_{K'=1}^{L} {L \choose K'} {M-L \choose K'} \exp\Bigg(- \frac{N-L}{4} \bigg(\frac{(K'){\mu}^2({\bf x}) - \delta'}{(K'){\mu}^2({\bf x}) + \nu^2}\bigg)^2 \Bigg) \nonumber
\end{align}

We will now show that the summation goes to $0$ as $M \to \infty$. We use the following bound 
\begin{equation} \label{eq-choose1}
\exp\bigg(K' \log\Big(\frac{L}{K'}\Big)\bigg) \leq {L \choose K'} \leq \exp\bigg(K' \log\Big(\frac{Le}{K'}\Big)\bigg) 
\end{equation}
to upper bound each term of summation, $s_{K'}$ by
\begin{align}
	s_{K'} &\leq \exp \Bigg(K' \log\Big(\frac{Le}{K'}\Big) + K' \log\Big(\frac{(M-L)e}{K'}\Big) - \frac{N-L}{4} \Big(\frac{K'{\mu}^2({\bf x}) - \delta'}{K'{\mu}^2({\bf x}) + \nu^2} \Big)^2 \Bigg) \nonumber\\
	&= \exp \Bigg(L\frac{K'}{L} \log \frac{e}{\frac{K'}{L}} + L\frac{K'}{L} \log \frac{(\beta - 1)e}{\frac{K'}{L}} - C_0 L \Big(\frac{L\frac{K'}{L}{\mu}^2({\bf x}) - \delta'}{L\frac{K'}{L}{\mu}^2({\bf x}) + \nu^2} \Big)^2  \Bigg)   \nonumber
\end{align}

We upper bound the whole summation by maximizing the function
\begin{align}
f(z) &=  Lz \log \frac{e}{z} + Lz \log \frac{(\beta - 1)e}{z} - C_0 L \Big(\frac{Lz{\mu}^2({\bf x}) - \delta'}{Lz{\mu}^2({\bf x}) + \nu^2} \Big)^2   \nonumber \\
		&= - 2Lz \log z + Lz (2 + \log(\beta - 1)) - C_0 L \Big(\frac{Lz{\mu}^2({\bf x}) - \zeta \mu^2({\bf x})}{Lz{\mu}^2({\bf x}) + \nu^2} \Big)^2 	\label{f-def}
\end{align} 
for $z \in [ \frac1L, 1]$. If $f(z)$ attains its maximum at $z_0$, we then have $$\sum_{K' = 1}^{L} s_{K'} \leq L\exp(f(z_0)).$$

For clarity of presentation, we will now state two technical lemmas.

\begin{lem} \label{pre-lem4-1}
Let $g(z)$ be a twice differentiable function on $[a,b]$ that has a continuous second derivative. If $g(a) < 0$, $g(b) < 0$, and $g'(a) < 0$, $g'(b) > 0$, and $g''(a) < 0$, $g''(b) < 0$, then $g''(x)$ is equal to 0 for at least two points in $[a, b]$.
\end{lem}

\begin{proof} Since $g'(a) < 0$ and $g'(b)>0$, $g'(x)$ has to be increasing in a subset $E \subset [a, b]$. Then $g''(x) > 0$ for some $x_0 \in E$. Since $g''(a) < 0$, $g''(x_0) > 0$ and $g''(x)$ is continuous, there exists $x_1 \in [a, x_0]$ such that $g''(x_1) = 0$. Similarly, since $g''(b) < 0$, there exists $x_2 \in [x_0, b]$ such that $g''(x_2) = 0$.
\end{proof}

\begin{lem} \label{pre-lem4-2}
Let $p(z) = a_4z^4 + a_3 z^3 + a_2 z^2 + a_1 z + a_0$ be a polynomial over ${\mathbb R}$ such that $a_4, a_3, a_0 > 0$. Then $p(z)$ can have at most two positive roots.
\end{lem}
\begin{proof}
Let $r_p^{(1)}, r_p^{(2)}, r_p^{(3)}, r_p^{(4)}$ be the roots of $p(z)$, counting multiplicities. Since $$r_p^{(1)} r_p^{(2)} r_p^{(3)} r_p^{(4)} = \frac{a_0}{a_4} > 0,$$ the number of positive roots must be even, and since
$$r_p^{(1)} + r_p^{(2)} + r_p^{(3)} + r_p^{(4)} = -\frac{a_3}{a_4} < 0,$$ not all the roots could be positive. The result follows. 
\end{proof}

\begin{lem} \label{lem4} For $L$ sufficiently large, $f(z)$ (see Equation (\ref{f-def})) is negative for all $z \in [\frac1L, 1]$. Moreover the endpoints of the interval, $z_0^{(1)} =\frac1L$ and $z_0^{(2)} = 1$ are its local maxima.
\end{lem}

\begin{proof}
We first confirm that $f(z)$ is negative at the endpoints of the interval. 
We use the notation $\overrightarrow{\approx}$ for denoting the behavior of $f(z)$ for large $L$, and $\prec$ and $\succ$ for inequialities that hold asymptotically.
\begin{equation}
	f\bigg(\frac1L\bigg) =  2\log L + 2 + \log(\beta - 1) - C_0 L \Big(\frac{{\mu}^2({\bf x}) (1 - \zeta)}{{\mu}^2({\bf x})  + \nu^2} \Big)^2 \prec 0 \label{f1L} 
\end{equation}
for sufficiently large $L$, since $L {\mu}^4({\bf x})$ grows faster than $\log L$. Also for large $L$, we have
\begin{align}
	f(1) &= L (2 + \log(\beta - 1)) -  C_0 L \Big(\frac{{\mu}^2({\bf x}) (L - \zeta)}{L{\mu}^2({\bf x})  + \nu^2} \Big)^2 \nonumber \\
		&\overrightarrow{\approx} L(2 + \log(\beta - 1) - C_0) \prec 0.  \label{f1v}
\end{align}

We now examine the derivative of $f(z)$, given by
\begin{align}	
	f'(z) 
	  &= -2L \log z + L \log(\beta -1) - 2C_0 L^2 \mu^4({\bf x}) (\nu^2 + \zeta \mu^2({\bf x})) \frac{Lz-\zeta}{(Lz{\mu}^2({\bf x}) + \nu^2)^3}  \nonumber
\end{align}
Also,
\begin{align}
	f'\bigg(\frac1L\bigg) &= 2L \log L + L \log(\beta - 1) - 2 C_0 L^2 \mu^4({\bf x}) (\nu^2 + \zeta \mu^2({\bf x})) \frac{1 - \zeta}{(\mu^2({\bf x}) + \nu^2)^3} \nonumber \\
	 &\overrightarrow{\approx} L \bigg(2 \log L + \log(\beta-1) - 2 \hat{C}_0 \frac{L\mu^4({\bf x})}{(\mu^2({\bf x}) + \nu^2)^2} \bigg) \prec 0   \nonumber
\end{align}
for sufficiently large $L$, since $L \mu^4({\bf x})$ grows faster than $\log L$. Similarly
\begin{align}
	f'(1) &= L \log(\beta -1) - 2C_0 L^2 \mu^4({\bf x}) (\nu^2 + \zeta \mu^2({\bf x})) \frac{L - \zeta}{(L\mu^2({\bf x}) + \nu^2)^3} \nonumber \\
		&\overrightarrow{\approx} L \log(\beta - 1) - 2C_0 \frac{1}{\mu^2({\bf x})} (\nu^2 + \zeta \mu^2({\bf x})) \succ  0   \nonumber
\end{align}
since $\frac{1}{\mu^2({\bf x})}$ grows slower than $\sqrt{\frac{L}{\log L}}$.

Additionally,
\begin{align}
	f''&(z) = -\frac{2L}{z} - 2C_0 L^2 \mu^4({\bf x}) (\nu^2 + \zeta \mu^2({\bf x})) \Big(\frac{-2Lz\mu^2({\bf x}) + \nu^2 + 3 \zeta \mu^2({\bf x})}{(Lz\mu^2({\bf x}) + \nu^2)^4} \Big)L \nonumber \\
				&= \frac{-2L}{z(Lz\mu^2({\bf x}) + \nu^2)^4} \Bigg( (Lz\mu^2({\bf x}) + \nu^2)^4 + C_0 L^2 \mu^4({\bf x}) (\nu^2 + \zeta \mu^2({\bf x})) (-2Lz\mu^2({\bf x}) + \nu^2 + 3 \zeta \mu^2({\bf x}))z\Bigg) \label{f2num}
\end{align}
Thus,
\begin{equation}
	f''\bigg(\frac1L\bigg) = -2L \Bigg(L + C_0 L^2 \mu^4({\bf x}) (\nu^2 + \zeta \mu^2({\bf x})) \frac{-2\mu^2({\bf x}) + \nu^2 + 3 \zeta \mu^2({\bf x})}{(\mu^2({\bf x}) + \nu^2)^4} \Bigg) \prec 0    \nonumber
\end{equation}
and
\begin{align}
	f''(1) &= -2L\Bigg(1 + C_0 L^2 \mu^4({\bf x}) (\nu^2 + \zeta \mu^2({\bf x})) \frac{-2L\mu^2({\bf x}) + \nu^2 + 3 \zeta \mu^2({\bf x})}{(L\mu^2({\bf x}) + \nu^2)^4} \Bigg) \nonumber \\
			&\overrightarrow{\approx} -2L \Bigg(1 - 2C_0 \frac{\nu^2 + \zeta \mu^2({\bf x})}{L \mu^2({\bf x})} \Bigg) \prec 0.  \nonumber
\end{align}

Since $f(z)$ is twice differentiable function on $[\frac1L, 1]$ with a continuous second derivative, Lemma \ref{pre-lem4-1} implies that $f''(z)$ crosses 0 at least twice in this interval. 
Next we examine the polynomial (see Equation (\ref{f2num})),
$$p(z) = (Lz\mu^2({\bf x}) + \nu^2)^4 + 2C_0 L^2 \mu^4({\bf x}) (\nu^2 + \zeta \mu^2({\bf x})) (-2Lz\mu^2({\bf x}) + \nu^2 + 3 \zeta \mu^2({\bf x}))z.$$
Since $p(z)$ satisfies the conditions of Lemma \ref{pre-lem4-2}, we conclude that it has at most two positive roots, and thus at most two roots of $p(z)$ can lie in $[\frac1L, 1]$. In other words $f''(z)$ can cross $0$ for $z \in [\frac1L, 1]$ at most twice. Combining this with the previous information, we conclude that $f''(z)$ crosses 0 exactly twice in this interval, and that $f'(z)$ crosses 0 only once, and this point is a local minima of $f(z)$. Thus the local maxima of $f(z)$ are the endpoints $z_0^{(1)} = \frac1L$ and $z_0^{(2)} = 1$.
\end{proof}

Thus we have, 
\begin{align}
	p_{\textrm{err}} ({\cal D}|{\bf x}) &\leq 2 \exp\Bigg(-\frac{\zeta^2 C_0}{\nu^2} \frac{L\mu^4({\bf x})}{\nu^2 + 2 \zeta \mu^2({\bf x})}\Bigg)  + \sum_{K=0}^{L-1} \exp(\max\{f(z_0^{(1)}), f(z_0^{(2)}) \}) \nonumber \\
		&= 2 \exp\Bigg(-\frac{\zeta^2 C_0}{\nu^2} \frac{L\mu^4({\bf x})}{\nu^2 + 2 \zeta \mu^2({\bf x})}\Bigg)  +  \exp\Bigg(\log L + \max\bigg\{f\bigg(\frac1L\bigg), f(1) \bigg\} \ \Bigg)  \nonumber
\end{align}

From Equations (\ref{f1L}) and (\ref{f1v}), it is clear that $\log (L) + \max\bigg\{f\bigg(\frac1L\bigg), f(1) \bigg\} \to -\infty$ as $L \to \infty$. Hence with the conditions of Theorem \ref{thm1}, $p_{\textrm{err}} ({\cal D}|{\bf x}) \to 0$ as $L \to \infty$.

\subsubsection{Proof of Theorem \ref{thm3} (Error Metric 2)} \label{sssec2}
For asymptotic reliable recovery with Error Metric 2, 
we require that ${\mathbb P}(E_{\cal J})$ goes to 0 for only $K \leq (1 - \alpha)L$ with $\alpha \in (0,1)$. By a re-examination of Equation (\ref{eq17}), we observe that the right hand side of
\begin{align}
	{\mathbb P}(E_{\cal J}) &\leq \exp\Bigg(- \frac{N-L}{4} \bigg(\frac{\alpha L{\mu}^2({\bf x}) - \delta'}{\alpha L{\mu}^2({\bf x}) + \nu^2}\bigg)^2 \Bigg)   \nonumber
\end{align}
converges to 0 asymptotically, even when $L\mu^2({\bf x})$ converges to a constant. In this case $P$ does not have to grow with $N$. We let $\delta > 0$ (and hence $\delta'$) be a constant, and let $N = (4\hat{C}_3 + 1) L$ for 
\begin{equation}  \label{c3'}
\hat{C}_3 > \beta \bigg(\frac{\alpha L \mu^2({\bf x}) + \nu^2}{\alpha L \mu^2({\bf x}) - \delta'} \bigg)^2.
\end{equation} 
Given the decay rate of $\mu^2({\bf x})$ and that $\delta'>0$ is arbitrary, we note that this constant only depends on $\alpha, \beta, \mu({\bf x})$ and $\nu$. Hence

\begin{align}
	p_{\textrm{err}} ({\cal D}|{\bf x}) &\leq {\mathbb P} (E_{\cal I}^C) + \sum_{K=0}^{(1 - \alpha)L} {L \choose L -K} {M-L \choose L-K} \exp\Bigg(- \frac{N-L}{4} \bigg(\frac{(L-K){\mu}^2({\bf x}) - \delta'}{(L-K){\mu}^2({\bf x}) + \nu^2}\bigg)^2 \Bigg) \nonumber \\
	&\leq  2\exp\Bigg(-\frac{\delta^2}{4\nu^4} \frac{4\hat{C}_3 + 1}{4\hat{C}_3 + \frac{2\delta}{\nu^2} (4\hat{C}_3+1)} N \Bigg) \nonumber \\
	& \quad \quad+ \sum_{K'=\alpha L}^{L} \exp \Bigg(L H\bigg(\frac{K'}{L}\bigg) + (M-L) H\bigg(\frac{K'}{M - L}\bigg) - \hat{C}_3 L \bigg(\frac{K'{\mu}^2({\bf x}) - \delta'}{K'{\mu}^2({\bf x}) + \nu^2}\bigg)^2   \Bigg),   \nonumber
\end{align}
where $H(a) = -a \log(a) - (1-a)\log(1-a)$ is the entropy function for $a \in [0,1]$. Since $K'$ is greater than a linear factor of $L$ and since $P$ is a constant, and using Equation (\ref{c3'}), we see $p_{\textrm{err}} ({\cal D}|{\bf x}) \to 0$ exponentially fast as $L \to \infty$. 

\subsubsection{Proof of Theorem \ref{thm5} (Error Metric 3)} \label{sssec3}
An error occurs for Error Metric 3 if $$\sum_{k \in {\cal I} \backslash {\cal J}} |x_k|^2 \geq \gamma P.$$ Thus we can bound the error event for ${\cal J}$ from Lemma \ref{lem1} as
\begin{equation}
{\mathbb P}(E_{\cal J}) \leq \exp\Bigg(- \frac{N-L}{4} \bigg(\frac{\gamma P - \delta'}{\gamma P + \nu^2} \bigg)^2 \Bigg)  \nonumber
\end{equation}

Let $\delta' > 0$ be a fraction of $\gamma P$. We denote the number of index sets ${\cal J} \subset \{0,1,\dots, M\}$ with $|{\cal J}| = L$ as $T_*$ and note that $T_* \leq {M \choose L}$. Thus,
\begin{align}
	p_{\textrm{err}} ({\cal D}|{\bf x}) &\leq 2\exp\Bigg(-\frac{\delta^2}{4\nu^4} \frac{N}{N-L + \frac{2\delta}{\nu^2}N} N \Bigg) + {M \choose L} \exp\Bigg(- \frac{N-L}{4} \bigg(\frac{\gamma P - \delta'}{\gamma P + \nu^2}\bigg)^2 \Bigg). \nonumber 
\end{align}
For $N > C_5 L$, a similar argument to that of Section \ref{sssec2} proves that $p_{\textrm{err}} ({\cal D}|{\bf x}) \to 0$ exponentially fast as $L \to \infty$, where $C_5$ depends only on $\beta, \gamma, P$ and $\nu$.


\section{Proofs of Converses} \label{sec:converse}

Throughout this section, we will write ${\bf x}$ for ${\bf x}^{(M)}$ whenever there is no ambiguity.

\subsection{Genie-Aided Decoding and Connection with Noisy Communication Systems}

Let the support of ${\bf x}$ be ${\cal I} = \{i_1, i_2, \dots, i_L \}$ with $i_1 < i_2 < \dots < i_L$. We assume a genie provides ${\bf x}_{\cal I} = (x_{i_1}, x_{i_2}, \dots, x_{i_L})^T$ 
to the decoder defined in Section \ref{sec:results}.

Clearly we have
\begin{equation} 
   p_{err} \geq p_{err}^{\textrm{genie}}  \nonumber
\end{equation}

\subsubsection{Proof of Theorem \ref{thm2} (Error Metric 1)} \label{sssec4}
We derive a lower bound on the probability of genie-aided decoding error for any decoder. 
Consider a Multiple Input Single Output (MISO) transmission model given by an encoder, a decoder and a channel. The channel is specified by ${\bf H} = [x_{i_1} x_{i_2} \dots x_{i_L}] = {\bf x}_{\cal I}^T$. The encoder, ${\mathfrak E}_1: \{0,1\}^M \to {\mathbb C}^{L \times N}$, maps one of the ${M \choose L}$ possible binary vectors of (Hamming) weight $L$ to a codeword in ${\mathbb C}^{L \times N}$. This codeword is then transmitted over the MISO channel in $N$ channel uses.
The decoder is a mapping ${\mathfrak D}_1: {\mathbb C}^N \to \{0,1\}^M$ such that its output $\hat{{\bf c}}$ has weight $L$.

Let ${\bf c} \in \{0, 1\}^M$ and supp$({\bf c}) = {\cal J} = \{j_1, j_2, \dots, j_L \}$ with $j_1 < j_2 < \dots < j_L$. Let ${\bf z}_k^{\cal J} = (a_{k, j_1}, a_{k, j_2}, \dots, a_{k, j_L})^T$, where $a_{m,n}$ is the $(m,n)^{\textrm{th}}$ term of ${\bf A}$. The codebook is specified by
$${\mathfrak C}_1 = \Bigg\{
\left( \begin{array}{c}
{\bf z}_1^{\cal J} \: \: {\bf z}_2^{\cal J} \:\:
\hdots \:\:
{\bf z}_N^{\cal J}
\end{array} \right)  \Bigg| {\cal J} \subset \{1,2,\dots, M\}, |{\cal J}| = L \Bigg\},$$
and has size ${M \choose L}$. 
The output of the channel, ${\bf y}$ is
\begin{equation}	
	y_k = {\bf H} {\bf z}_k^{\cal I} + n_k \quad \textrm{ for } \quad k = 1,2, \dots, N,   \nonumber
\end{equation}
where $y_k$ and $n_k$ are the $k^{\textrm{th}}$ coordinates of ${\bf y}$ and ${\bf n}$ respectively. The average signal power is ${\mathbb E}(||{\bf z}_k^{\cal J}||^2) = L$, and the noise variance is ${\mathbb E}{n_k^2} = \nu^2$. The capacity of this channel in $N$ channel uses (without channel knowledge at the transmitter) is given by \cite{Vuc} $$C_{MISO} = N\log\bigg(1 + \frac1L \frac{{\mathbb E}(||{\bf z}_k^{\cal J}||^2)}{{\mathbb E}n_k^2} {\bf HH^{\dagger}} \bigg) = N\log\bigg(1 + \frac{P}{\nu^2}\bigg). $$

After $N$ channel uses, $p_{err}^{MISO} > 0$ if $\log{M \choose L} > C_{MISO}$. Using 
\begin{equation}	\label{eq-choose2}
\frac{1}{M+1} \exp\bigg(MH\bigg(\frac{L}{M}\bigg)\bigg) \leq {M \choose L} \leq  \exp\bigg(MH\bigg(\frac{L}{M}\bigg)\bigg),
\end{equation}
 we obtain the equivalent condition
\begin{equation}
	N < \frac{1}{\log\big(1 + \frac{P}{\nu^2} \big)} M H\bigg(\frac{1}{\beta}\bigg) - o(M),  \nonumber
\end{equation}
where $L = \beta M$, and $H(\cdot)$ is the entropy function. 

To prove Corollary \ref{coro4}, we first show that with high probability, all codewords of a Gaussian codebook satisfy a power constraint. Combining this with the strong converse of the channel coding theorem will complete the proof \cite{Gallager}. 
If ${\bf A}$ is chosen from a Gaussian distribution, then by Inequality (\ref{eq-dev2}), 
\begin{equation}
	{\mathbb P}\Bigg(\frac{1}{L}||{\bf z}_k^{\cal J}||^2 > \bigg(1 + 2 \bigg(\sqrt{\beta H\bigg(\frac{1}{\beta} \bigg) + \xi}\bigg) + 2 \bigg(\beta H\bigg(\frac{1}{\beta} \bigg) + \xi \bigg)  \Bigg) \leq \exp \Bigg(- \bigg(\beta H\bigg(\frac{1}{\beta}\bigg) + \xi \bigg) L \Bigg)   \nonumber
\end{equation}
for any ${\cal J} \subset \{1,2,\dots, M\}, |{\cal J}| = L $ and for $k= 1,2, \dots, N$. Let $\Lambda = 2 \sqrt{\beta H\big(\frac{1}{\beta}\big) + \xi} +2 \big(\beta H\big(\frac{1}{\beta} \big) + \xi \big) $ for $\xi>0$. By the union bound over all ${M \choose L}$ possible index sets ${\cal J}$ and $k= 1,2, \dots, N$,
\begin{equation}
	{\mathbb P}\bigg(\frac{1}{L}||{\bf z}_k^{\cal J}||^2 < \big(1 + \Lambda  \big), \:\: \forall {\cal J}, \: \: k=1,\dots, N\bigg) \geq 1 - N\exp \big(- \xi L \big).  \nonumber
\end{equation}
If the power constraint is satisfied, then the strong converse of the channel coding theorem implies that $p_{err} ({\bf A}|{\bf x})$ goes to 1 exponentially fast in $M$ if 
$$ N \prec \frac{1}{\log\bigg(1 + \frac{P(1+\Lambda)}{\nu^2} \bigg)} M H\bigg(\frac{1}{\beta}\bigg).$$

\subsubsection{Proof of Theorem \ref{thm4} (Error Metric 2)} \label{sssec5}

For any given ${\bf x}$ with $||{\bf x}||_0 = L$, we will prove the contrapositive. Let $P_{e_2}^{(M)}$ denote the probability of error with respect to Error Metric 2 for ${\bf x} \in {\mathbb C}^M$. 
We show that $N \succ C_4 L$ if $P_{e_2}^{(M)} \to 0$.

Consider a single input single output system, ${\cal S}$, whose input is ${\bf c} \in \{0, 1\}^M$, and whose output is $\hat{{\bf c}} \in \{0, 1\}^M$, such that $||{\bf c}||_0 = ||\hat{{\bf c}}||_0 = L$, and $||{\bf c} - \hat{{\bf c}}||_ 0 \leq 2 \alpha L$. The last condition states that the support of ${\bf c}$ and that of $\hat{{\bf c}}$ overlap in more than $(1 - \alpha) L$ locations, i.e. $P_{e_2}^{(M)} = 0$.
We are interested in the rates at which one can communicate reliably over ${\cal S}$. 

In our case $d({\bf c}, \hat{{\bf c}}) = \frac1M \sum_{k=1}^M d_H(c_i, \hat{c}_i)$, where ${\bf c}$ is i.i.d. distributed among ${M \choose L}$ binary vectors of length $M$ and weight $L$, and $d_H(\cdot,\cdot)$ is the Hamming distance. Thus $D \leq \frac{2 \alpha L}{M} = \frac{2\alpha}{\beta}$. 
We also note that ${\cal S}$ can be viewed as consisting of an encoder ${\mathfrak E}_1$, a MISO channel and a decoder, ${\mathfrak D}_1$ as described in Section \ref{sssec4}. Since the source is transmitted within distortion $\frac{2 \alpha}{\beta}$ over the MISO channel, we have \cite{Berger}
$$R\bigg(\frac{2 \alpha}{\beta}\bigg) < C_{MISO}. $$
In order to bound $R\big(\frac{2 \alpha}{\beta}\big)$, we first state a technical lemma.

\begin{lem} \label{lem4-1}
Let $\alpha \in (0, 1]$ and $\beta > 2$, and let
\begin{align} 
c(z) &= H(z) + (\beta -1) H\bigg(\frac{z}{\beta -1}\bigg) \nonumber \\ &= -2z \log(z) - (1-z) \log(1-z) + (\beta - 1) \log(\beta -1) - (\beta - 1 - z) \log (\beta - 1 -z), \nonumber
\end{align}
where $H(\cdot)$ is the entropy function. Then for $z \in [0, \alpha]$, $c(z) \geq 0$, and $c(z)$ attains its maximum at $z = \min\big(a, \frac{\beta - 1}{\beta} \big)$.
\end{lem}
\begin{proof}
By definition of $H(\cdot)$, $c(z) \geq 0$ for $z \in [0, \alpha]$. By examining
$$c'(z) =  -2 \log(z) + \log(1-z) + \log(\beta - 1 -z) = \log\bigg(\frac{(1-z)(\beta - 1-z)}{z^2}\bigg),$$
it is easy to see that $c'(z) \geq 0$ for $z \in \bigg(0, \min\big(\alpha, \frac{\beta -1}{\beta}\big) \bigg]$ and $c'(z) < 0$ otherwise.
\end{proof}

Thus we have
\begin{align}
	I({\bf c}&, \hat{{\bf c}}) \Big|_{||{\bf c}||_0 = ||\hat{{\bf c}}||_0 = L, ||{\bf c} - \hat{{\bf c}}||_0 \leq 2 \alpha L} = H({\bf c}) - H({\bf c} \: | \: \hat{{\bf c}}) \Big|_{||{\bf c}||_0 = ||\hat{{\bf c}}||_0 = L, ||{\bf c} - \hat{{\bf c}}||_ 0 \leq 2 \alpha L} \nonumber \\
			 &\geq \log {M \choose L} - \log\bigg(\sum_{K=0}^{\alpha L} {L \choose K}{M-L \choose K}\bigg) \nonumber \\
			 &\geq M H\bigg(\frac{1}{\beta}\bigg) - \log(M+1) - \log\Bigg(\sum_{K=0}^{\alpha L} \exp\bigg( L H \bigg(\frac{K}{L}\bigg) + (M - L) H\bigg(\frac{K}{M-L} \bigg) \bigg) \Bigg) \nonumber \\
			 &\geq \left\{ \begin{array}{ll}
MH\bigg(\frac{1}{\beta}\bigg) - \log(M+1) - \log(\alpha L + 1) - L\bigg(H(\alpha) + (\beta -1) H\bigg(\frac{\alpha}{\beta -1} \bigg) \bigg) & \textrm{if $\alpha\leq \frac{\beta -1}{\beta}$}\\
0 & \textrm{if $\alpha> \frac{\beta -1}{\beta}$} \nonumber 
\end{array} \right. ,
\end{align}
where the first inequality follows since given $\hat{{\bf c}}$, ${\bf c}$ is among $\sum_{K=0}^{\alpha L} {L \choose K}{M-L \choose K}$ possible binary vectors within Hamming distance $2 \alpha L$ from $\hat{{\bf c}}$. The second inequality follows from Inequality (\ref{eq-choose2}), and the third inequality follows by Lemma \ref{lem4-1}.

Thus $R\big(\frac{2 \alpha}{\beta}\big) \geq L C_{\alpha, \beta} - o(L)$, where 
\begin{equation} \label{Cabeta}
C_{\alpha, \beta} = \left\{ \begin{array}{ll}
\beta H\bigg(\frac{1}{\beta}\bigg) -H(\alpha) - (\beta -1) H\bigg(\frac{\alpha}{\beta -1}\bigg) & \textrm{if $\alpha\leq \frac{\beta -1}{\beta}$}\\
0 & \textrm{if $\alpha> \frac{\beta -1}{\beta}$} 
\end{array} \right.
\end{equation}

Therefore if $P_{e_2}^{(M)} = 0$, then $$LC_{\alpha, \beta} - o(L) < N \log\bigg(1 + \frac{P}{\nu^2} \bigg) $$
or equivalently for large $M$,
$$ N \succ \frac{C_{\alpha, \beta}}{\log\bigg(1 + \frac{P}{\nu^2} \bigg)} L.$$
The contrapositive statement proves Theorem \ref{thm4}.

\subsubsection{Proof of Theorem \ref{thm6} (Error Metric 3)} \label{sssec6}

For Error Metric 3, we assume that $\rho({\bf x}) = \max_{i \in {\cal I}} |x_i|$ and $\mu({\bf x}) = \min_{i \in {\cal I}} |x_i|$ both decay at rate $O\big(\sqrt{\frac1L}\big)$. Thus $P$ is constant. In the absence of this assumption, some terms of ${\bf x}$ can be asymptotically dominated by noise. Such terms are unimportant for recovery purposes, and therefore could be replaced by zeros (in the definition of ${\bf x}$) with no significant harm.

Let $\alpha(\gamma, {\bf x}) = \min \big(\frac{\gamma P}{L \mu^2({\bf x})}, 1\big)$. Let $P_{e_3}^{(M)}$ denote the probability of error with respect to Error Metric 3 for ${\bf x} \in {\mathbb C}^M$. If $P_{e_3}^{(M)} = 0$ and if an index set ${\cal J}$ is recovered, then $\sum_{k \in {\cal I} \backslash {\cal J}} |x_k|^2 \leq \gamma P$, where ${\cal I} = \textrm{supp}({\bf x})$. This implies that $|{\cal I} \backslash {\cal J}| \leq \alpha(\gamma, {\bf x}) L$. Thus $P_{e_3}^{(M)} = 0$ implies that $P_{e_2}^{(M)} = 0$ when recovering $\alpha(\gamma, {\bf x})$ fraction of the support of ${\bf x}$. 
As shown in Section \ref{sssec5}, reliable recovery of ${\bf x}$ is not possible if 
$$ N \prec \frac{C_{\alpha(\gamma, {\bf x}), \beta}}{{\log\bigg(1 + \frac{P}{\nu^2} \bigg)}} L,$$
where $C_{\alpha(\gamma, {\bf x}), \beta}$ is a constant (as defined in Equation (\ref{Cabeta})) that only depends on $\gamma, \beta, \mu({\bf x})$ and $P$ for a given ${\bf x}$.


\section{Sublinear Regime} \label{sec:sublinear}

For completeness, we also state the equivalent theorems, when $L = o(M)$. The proofs follow the same steps as those in the linear regime. For the proofs of converse results, we use the bounds from Equation (\ref{eq-choose1}) instead of those of Equation (\ref{eq-choose2}).

\begin{thm} \label{thm5.1} (Achievability for Error Metric 1) 
Let a sequence of sparse vectors, $\{{\bf x}^{(M)} \in {\mathbb C}^M\}_M$ with $||{\bf x}^{(M)}||_0 = L = o(M)$ be given. Then asymptotic reliable recovery is possible for $\{{\bf x}^{(M)}\}$ with respect to Error Metric 1 if $L {\mu}^4({\bf x}^{(M)}) \to \infty$ as $L \to \infty$ and
\begin{equation}	
	  N \succ  C'_1 \: L \log(M-L)
\end{equation}
for some constant $C'_1>0$ that depends only on ${\mu}({\bf x}^{(M)})$ and $\nu$.
\end{thm}

\begin{proof} The proof is similar to that of Theorem \ref{thm1}, with $f(z)$ replaced by
$$ k(z) = -2Lz \log z + 2Lz + Lz \log\bigg(\frac{M-L}{L}\bigg) - \frac{N-L}{4} \Big(\frac{Lz{\mu}^2({\bf x}) - \zeta \mu^2({\bf x})}{Lz{\mu}^2({\bf x}) + \nu^2} \Big)^2.$$
The behavior of $k(z), k'(z)$ and $k''(z)$ at the endpoints $\{\frac1L, 1\}$, is the same as that in the proof of Theorem \ref{thm1} whenever $N = C_1' L \log(M-L)$. The result follows.
\end{proof}

\begin{thm} \label{thm5.2} (Converse for Error Metric 1) 
Let a sequence of sparse vectors, $\{{\bf x}^{(M)} \in {\mathbb C}^M\}_M$ with $||{\bf x}^{(M)}||_0 = L = o(M)$ be given. Then asymptotic reliable recovery is not possible for $\{{\bf x}^{(M)}\}$ with respect to Error Metric 1 if 
\begin{equation}	
	  N \prec  C'_2 \frac{L \log(M-L)}{\log P}
\end{equation}
for some constant $C'_2>0$ that depends only on $P$ and $\nu$.
\end{thm}
\begin{proof} The proof is similar to that of Theorem \ref{thm2}. \end{proof}

\begin{thm} \label{thm5.3} (Achievability for Error Metric 2) 
Let a sequence of sparse vectors, $\{{\bf x}^{(M)} \in {\mathbb C}^M\}_M$ with $||{\bf x}^{(M)}||_0 = L = o(M)$ be given such that $L\mu^2({\bf x}^{(M)})$ and $P$ are constant. Then asymptotic reliable recovery is possible for $\{{\bf x}^{(M)}\}$ with respect to Error Metric 2 if
\begin{equation}	
	  N \succ  C'_3 \: L \log(M-L)
\end{equation}
for some constant $C'_3>0$ that depends only on $\alpha$, ${\mu}({\bf x}^{(M)})$ and $\nu$. 
\end{thm}
\begin{proof} The proof is similar to that of Theorem \ref{thm3}. \end{proof}

\begin{thm} \label{thm5.4} (Converse for Error Metric 2) 
Let a sequence of sparse vectors, $\{{\bf x}^{(M)} \in {\mathbb C}^M\}_M$ with $||{\bf x}^{(M)}||_0 = L = o(M)$ be given such that $P$ is constant. Then asymptotic reliable recovery is not possible for $\{{\bf x}^{(M)}\}$ with respect to Error Metric 2 if 
\begin{equation}	
	  N \prec  C'_4 \: L \log(M-L)
\end{equation}
for some constant $C'_4 > 0$ that depends only on $\alpha, P$ and $\nu$.
\end{thm}
\begin{proof} We have the following technical lemma,

\begin{lem} \label{lem5-1}
Let $\alpha \in (0, 1]$ and $L =o(M)$, and let
$$d(z) = 2z - 2z \log(z) + z \log\Big(\frac{M-L}{L} \Big).$$
Then for $z \in [0, \alpha]$, and for sufficiently large $M$, $d(z)$ attains its maximum at $z = \alpha$.
\end{lem}
\begin{proof}
By examining
$$d'(z) = -2 \log(z) + \log\Big(\frac{M-L}{L} \Big) = \log\Big(\frac{M-L}{Lz^2} \Big),$$
it is easy to see that $d'(z) \succ 0$ for sufficiently large $M$. 
\end{proof}

\emph{Continuation of the proof of the theorem: } Thus we have,
\begin{align}
	I({\bf c}&, \hat{{\bf c}}) \Big|_{||{\bf c}||_0 = ||\hat{{\bf c}}||_0 = L, ||{\bf c} - \hat{{\bf c}}||_0 \leq 2 \alpha L} = H({\bf c}) - H({\bf c} \: | \: \hat{{\bf c}}) \Big|_{||{\bf c}||_0 = ||\hat{{\bf c}}||_0 = L, ||{\bf c} - \hat{{\bf c}}||_ 0 \leq 2 \alpha L} \nonumber \\
			 &\geq L \log\bigg(\frac{M}{L}\bigg) - \log\Bigg(\sum_{K=0}^{\alpha L} \exp\bigg( K \log\bigg(\frac{Le}{K}\bigg) + K \log\bigg(\frac{(M-L)e}{K} \bigg) \bigg) \Bigg) \nonumber \\
			 &\geq L \log(M) - \alpha L \log(M-L) - o(L\log M) \geq (1 - \alpha) L \log(M-L) - o(L\log M)\nonumber,
\end{align}
where the first inequality follows 
from Inequality (\ref{eq-choose1}), and the second inequality follows by Lemma \ref{lem5-1} for sufficiently large $M$. The rest of the proof is analogous to that of Theorem \ref{thm4}.\end{proof}

\begin{thm} \label{thm5.5} (Achievability for Error Metric 3) 
Let a sequence of sparse vectors, $\{{\bf x}^{(M)} \in {\mathbb C}^M\}_M$ with $||{\bf x}^{(M)}||_0 = L = o(M)$ be given such that $P$ is constant. Then asymptotic reliable recovery is possible for $\{{\bf x}^{(M)}\}$ with respect to Error Metric 3 if
\begin{equation}	
	  N \succ  C'_5 \: L \log(M-L)
\end{equation}
for some constant $C'_5>0$ that depends only on $\gamma$, $P$ and $\nu$. 
\end{thm}
\begin{proof} The proof is similar to that of Theorem \ref{thm5}. \end{proof}
%

\begin{thm} \label{thm5.6} (Converse for Error Metric 3) 
Let a sequence of sparse vectors, $\{{\bf x}^{(M)} \in {\mathbb C}^M\}_M$ with $||{\bf x}^{(M)}||_0 = L = o(M)$ be given such that $P$ is constant and the non-zero terms decay to zero at the same rate. Then asymptotic reliable recovery is not possible for $\{{\bf x}^{(M)}\}$ with respect to Error Metric 3 if 
\begin{equation}	
	  N \prec  C'_6 \: L \log(M-L)
\end{equation}
for some constant $C'_6\geq0$ that depends only on $\gamma, P, \mu({\bf x}^{(M)})$ and $\nu$.
\end{thm}
\begin{proof} 
As in the proof of Theorem \ref{thm6}, we let $\alpha(\gamma, {\bf x}) = \min \big(\frac{\gamma P}{L \mu^2({\bf x})}, 1\big)$, and conclude that $P_{e_3}^{(M)} = 0$ implies that $P_{e_2}^{(M)} = 0$ when recovering $\alpha(\gamma, {\bf x})$ fraction of the support of ${\bf x}$. 
The rest of the proof is analogous to that of Theorem \ref{thm5.4}.
 \end{proof}







\end{document}